\documentclass[letterpaper, 10 pt, conference]{ieeeconf}
\IEEEoverridecommandlockouts
\usepackage{cite}
\usepackage{amsmath,amssymb,amsfonts}
\usepackage{algorithmic}
\usepackage{graphicx}
\usepackage{textcomp}
\usepackage{xcolor}
\usepackage{mathrsfs}
\usepackage{psfrag}
\usepackage{cancel}
\usepackage{verbatim}

\usepackage{circuitikz}
\usetikzlibrary{decorations.pathreplacing,calligraphy,patterns}

\makeatletter
\let\NAT@parse\undefined
\makeatother
\usepackage{hyperref}
\usepackage{theoremref}

\def\mc{\mathcal}
\def\mb{\mathbb}

\begin{document}
\sloppy
\title{\LARGE \bf Hybrid integrator-gain system based integral resonant controllers for negative imaginary systems}
\author{Kanghong Shi,$\quad$Ian R. Petersen, \IEEEmembership{Life Fellow, IEEE} 
\thanks{This work was supported by the Australian Research Council under grant DP230102443.}
\thanks{K. Shi and I. R. Petersen are with the School of Engineering, College of Engineering, Computing and Cybernetics, Australian National University, Canberra, Acton, ACT 2601, Australia. K. Shi is also with the Australian Centre for Robotics, University of Sydney, NSW 2006, Australia.
        {\tt kanghong@ieee.org}, {\tt ian.petersen@anu.edu.au}.}%
}

\newtheorem{definition}{Definition}
\newtheorem{theorem}{Theorem}
\newtheorem{conjecture}{Conjecture}
\newtheorem{lemma}{Lemma}
\newtheorem{remark}{Remark}
\newtheorem{corollary}{Corollary}
\newtheorem{assumption}{Assumption}

\maketitle
\thispagestyle{plain}
\pagestyle{plain}

\begin{abstract}
We introduce a hybrid control system called a hybrid integrator-gain system (HIGS) based integral resonant controller (IRC) to stabilize negative imaginary (NI) systems. A HIGS-based IRC has a similar structure to an IRC, with the integrator replaced by a HIGS. We show that a HIGS-based IRC is an NI system. Also, for a SISO NI system with a minimal realization, we show there exists a HIGS-based IRC such that their closed-loop interconnection is asymptotically stable. Also, we propose a proportional-integral-double-integral resonant controller ($\text{PII}^2\text{RC}$) and a HIGS-based $\text{PII}^2\text{RC}$, and we show that both of them can be applied to asymptotically stabilize an NI system. 
An example is provided to illustrate the proposed results.
\end{abstract}

\begin{keywords}
hybrid integrator-gain system, integral resonant control (IRC), negative imaginary (NI) system, stability, robust control.
\end{keywords}

\section{INTRODUCTION}
Negative imaginary (NI) systems theory was introduced in \cite{lanzon2008stability,petersen2010feedback} to address the robust control problem for flexible structures \cite{preumont2018vibration,halim2001spatial,pota2002resonant}, which usually have highly resonant dynamics. Roughly speaking, a square, real-rational and proper transfer function $G(s)$ is said to be NI if it has no strict right-half plane poles and its frequency response $G(j\omega)$ satisfies $j\left[G(j\omega)-G(j\omega)^*\right]\geq 0$ for all $\omega>0$ \cite{lanzon2008stability}. Typical examples of NI systems are mechanical systems with colocated force actuators and position sensors. NI systems theory provides an alternative approach to the passivity theory \cite{brogliato2007dissipative} when velocity measurements are unavailable. In comparison to passivity theory, which can only deal with systems with a relative degree of zero or one, an advantage of NI systems theory is that it allows systems to have relative degree zero, one, and two \cite{shi2024necessary}. An NI system can be stabilized using a strictly negative imaginary (SNI) controller. Under some assumptions, the positive feedback interconnection of an NI system $G(s)$ and an SNI system $R(s)$ is asymptotically stable if and only if the DC loop gain of the interconnection is strictly less than unity; i.e., $\lambda_{max}(G(0)R(0))<1$ (e.g., see \cite{lanzon2017feedback}). NI systems theory has found its applications in many fields including nano-positioning \cite{mabrok2013spectral,das2014mimo,das2014resonant,das2015multivariable}, control of lightly damped structures \cite{cai2010stability,rahman2015design,bhikkaji2011negative}, and control of power systems \cite{chen2023nonlinear}.

NI systems theory was extended to nonlinear systems in \cite{ghallab2018extending,shi2021robust,shi2023output}. A nonlinear system is said to be NI if the system is dissipative with respect to the inner product of the system input and the time derivative of the system output. Under some assumptions, a nonlinear NI system can be asymptotically stabilized using a nonlinear output strictly negative imaginary (OSNI) system. Such a nonlinear extension of NI systems theory not only makes NI systems theory applicable to a broader class of plants, but also allows the use of more advanced controllers.

One such controller is the hybrid integrator-gain system (HIGS). HIGS elements were introduced in \cite{deenen2017hybrid} to overcome the inherent limitations of linear control systems (see e.g., \cite{middleton1991trade}). A HIGS switches between an integrator mode and a gain mode in order to have a sector-bounded input-output relationship. Compared with an integrator, which has a phase lag of $90^\circ$, a HIGS has a similar magnitude slope but only a $38.1^\circ$ phase lag. This $52.9^\circ$ phase reduction can significantly reduce time delay, and as a consequence, the overshoot; see \cite{van2020hybrid} for a concrete example; see also \cite{van2021overcoming,heertjes2023overview,deenen2021projection,van2023small}.
It is shown in \cite{shi2022negative} that a HIGS element is a nonlinear NI system, and can be applied as a controller to asymptotically stabilize an NI plant. The paper \cite{shi2023MEMS} proposes a control methodology for multi-input multi-output (MIMO) NI systems using multi-HIGS controllers. Also, \cite{shi2023MEMS} reports the results of a hardware experiment where a multi-HIGS is applied to improve the performance of a micro-electromechanical system (MEMS) nanopositioner. The HIGS-based control methodologies proposed in \cite{shi2022negative,shi2023MEMS} also motivated an NI systems theory for systems with switching \cite{shi2023nonlinear}. In addition, based on the discrete-time NI systems theory \cite{shi2023discrete}, a digital control approach is proposed for NI systems \cite{shi2024digital}, where discrete-time HIGS are used as controllers.
Considering the effectiveness of HIGS-based control and the advantages of HIGS elements over linear integral controllers, we may naturally ask the question: rather than using a HIGS as a standalone controller, can we replace integrators with HIGS elements in more intricate controllers to enhance control performance? In this paper, we investigate this problem for a HIGS-based integral resonant controller .

IRC was introduced in \cite{aphale2007integral} to provide damping control for flexible structures. For a system with transfer matrix $G(s)$, an IRC is implemented by first adding a direct feedthrough $D$ to the system $G(s)$ and then applying an integrator in positive feedback to $G(s)+D$. Adding such a feedthrough $D$ changes the pole-zero interlacing of $G(s)$ into a zero-pole interlacing in $G(s)+D$. It is shown in \cite{petersen2010feedback,bhikkaji2008multivariable} that an IRC is an SNI system and can stabilize systems with the NI property. Since IRC are effective in damping control and easy to implement, they have been widely applied in the control of NI systems; e.g., see \cite{yue2015integral,bhikkaji2008integral,russell2017evaluating}.

In this paper, we propose a HIGS-based IRC by replacing the integrator in an IRC by a HIGS element. The advantages of a HIGS-based IRC are two-fold: 1) it utilizes the advantages of a HIGS in terms of small phase lag, reduced time delay and reduced overshoot; 2) A HIGS has two parameters -- the integrator frequency and gain value, while an integrator only has one parameter $\Gamma$. Hence, a greater degree of freedom in parameters is allowed in controller design using a HIGS-based IRC. We provide a state-space model of a HIGS-based IRC. We show that a HIGS-based IRC has the nonlinear NI property. We also show that given a SISO NI system with minimal realization, there exists a HIGS-based IRC such that their closed-loop interconnection is asymptotically stable. This is illustrated using an example.

We also investigate proportional-integral-double-integral resonant controllers ($\text{PII}^2\text{RC}$) in this paper. A $\text{PII}^2\text{RC}$ is implemented by replacing the integrator in an IRC by a proportional-integral-double-integral controller with the transfer function $C(s)=k_p+k_1/s+k_2/s^2$. We show that a $\text{PII}^2\text{RC}$ is an SNI system and can asymptotically stabilize an NI plant. Then, by replacing the integrators in an $\text{PII}^2\text{RC}$ with HIGS elements, we construct a HIGS-based $\text{PII}^2\text{RC}$. We show that a HIGS-based $\text{PII}^2\text{RC}$ can also provide asymptotic stabilization for NI plants.

The rest of the paper is organized as follows: Section \ref{sec:pre} provides some preliminary results on negative imaginary systems theory, IRC and HIGS. Section \ref{sec:HIGS IRC} provides a model for the HIGS-based IRC and shows that it can be used in the control of NI systems. Section \ref{sec:PII2RC} introduces a $\text{PII}^2\text{RC}$ and also gives a stability proof for the interconnection of an NI system a $\text{PII}^2\text{RC}$. Section \ref{sec:HIGS PII2RC} introduces the HIGS-based $\text{PII}^2\text{RC}$ and shows that it can be applied in the stabilization of NI systems. In Section \ref{sec:example}, we illustrate the main results in this paper that are given in Section \ref{sec:HIGS IRC} on a mass-spring system example. The paper is concluded in Section \ref{sec:conclusion}.

Notation: The notation in this paper is standard. $\mathbb R$ denotes the field of real numbers. $\mathbb R^{m\times n}$ denotes the space of real matrices of dimension $m\times n$. $A^T$ denotes the transpose of a matrix $A$.  $A^{-T}$ denotes the transpose of the inverse of $A$; that is, $A^{-T}=(A^{-1})^T=(A^T)^{-1}$. $\lambda_{max}(A)$ denotes the largest eigenvalue of a matrix $A$ with real spectrum. $\|\cdot\|$ denotes the standard Euclidean norm. For a real symmetric or complex Hermitian matrix $P$, $P>0\ (P\geq 0)$ denotes the positive (semi-)definiteness of a matrix $P$ and $P<0\ (P\leq 0)$ denotes the negative (semi-)definiteness of a matrix $P$. A function $V: \mb R^n \to \mb R$ is said to be positive definite if $V(0)=0$ and $V(x)>0$ for all $x\neq 0$.

\section{PRELIMINARIES}\label{sec:pre}
\subsection{Negative imaginary systems}
We consider systems of the form
\begin{subequations}\label{eq:general nonlinear system}
\begin{align}
	\dot{x} =&\ f(x,u),\\
	y =&\ h(x),
\end{align}
\end{subequations}
where $x\in \mathbb R^n$, $u,y\in \mathbb R^p$ are the state, input and output of the system, respectively. Here, $f:\mathbb R^n\times \mathbb R^p \to \mathbb R^n$ is a Lipschitz continuous function and $h:\mathbb R^n \to \mathbb R^p$ is a continuously differentiable function. We assume $f(0,0)=0$ and $h(0)=0$.
\begin{definition}[NI systems]\cite{ghallab2018extending,shi2021robust}\label{def:nonlinear NI}
	A system of the form (\ref{eq:general nonlinear system}) is said to be a negative imaginary (NI) system if there exists a positive definite continuously differentiable storage function $V:\mathbb{R}^n \to \mathbb{R}$ such that for any locally integrable input $u$ and solution $x$ to (\ref{eq:general nonlinear system}),
	\begin{equation*}\label{eq:Lyapunov_dissipative_nonlinear_NI}
		\dot{V}(x(t)) \leq u(t)^T\dot{y}(t), \quad \forall\, t \geq 0.
	\end{equation*}
\end{definition}
We provide the following conditions for a linear system to be NI. This condition is referred to as the NI lemma (see \cite{xiong2010negative}).
\begin{lemma}[NI lemma]\cite{xiong2010negative}\label{lemma:NI}
Let $(A,B,C,D)$ be a minimal state-space realization of an $p\times p$ real-rational proper transfer function matrix $G(s)$ where $A\in \mathbb R^{n\times n}$, $B\in \mathbb R^{n\times p}$, $C\in \mathbb R^{p\times n}$, $D\in \mathbb R^{p\times p}$. Then $G(s)$ is NI if and only if:

1. $\det A \neq 0$, $D=D^T$;

2. There exists a matrix $Y=Y^T>0$, $Y\in \mathbb R^{n\times n}$ such that
\begin{equation}\label{eq:NI lemma}
	AY+YA^T\leq 0, \textnormal{ and }\  B+AYC^T=0.
\end{equation}	
\end{lemma}

\subsection{Integral resonant control}
The implementation of an IRC is shown in Fig.~\ref{fig:CT_IRC}. Given a SISO plant with a transfer function $G(s)$, we apply a direct feedthrough $D$ and also an integral controller
\begin{equation}\label{eq:CT integrator}
	C(s)=\frac{\Gamma}{s}
\end{equation}
in positive feedback with $G(s)+D$. Here, we require the matrices $\Gamma, D\in \mb R$ to satisfy $D<0$ and $\Gamma>0$. The block diagram in Fig.~\ref{fig:CT_IRC} can be equivalently represented by the block diagram in Fig.~\ref{fig:CT_IRC_equivalence}, where $K(s)$ is
given as
\begin{equation}\label{eq:K(s)}
	K(s)=\frac{C(s)}{1-C(s)D}.
\end{equation}
Substituting (\ref{eq:CT integrator}) in (\ref{eq:K(s)}), we obtain the transfer function of the IRC:
\begin{equation}\label{eq:CT IRC}
	K(s)=\frac{\Gamma}{s-\Gamma D}.
\end{equation}
An IRC is an SNI system, and can be used in the control of NI plants (see \cite{bhikkaji2008multivariable,lanzon2008stability,petersen2010feedback}).

\begin{figure}[h!]
\centering
\psfrag{r}{$r$}
\psfrag{e_s}{$\overline e(s)$}
\psfrag{U_s}{$U(s)$}
\psfrag{Y_s}{$Y(s)$}
\psfrag{barY_s}{$\overline Y(s)$}
\psfrag{G_s}{$G(s)$}
\psfrag{C_s}{$C(s)$}
\psfrag{D}{$D$}
\includegraphics[width=8.5cm]{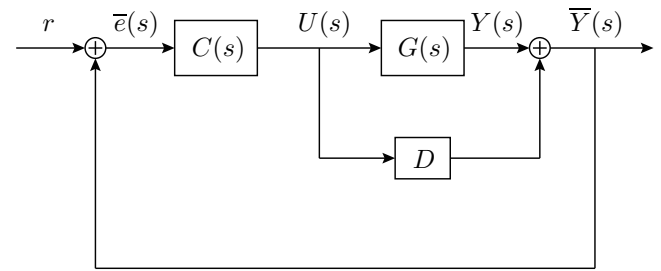}
\caption{Closed-loop interconnection of an integrator $C(s)=\frac{\Gamma}{s}$ and $G(s)+D$.}
\label{fig:CT_IRC}
\end{figure}

\begin{figure}[h!]
\centering
\psfrag{r}{$r$}
\psfrag{e_s}{$e(s)$}
\psfrag{U_s}{$U(s)$}
\psfrag{Y_s}{\hspace{0.5cm}$Y(s)$}
\psfrag{G_s}{$G(s)$}
\psfrag{K_s}{$K(s)$}
\includegraphics[width=8.5cm]{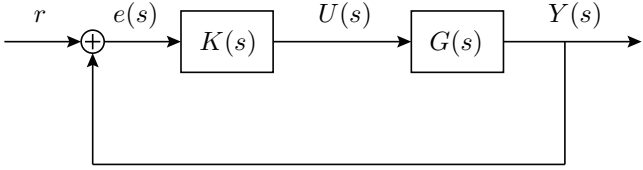}
\caption{Closed-loop interconnection of an IRC and a plant. This is equivalent to the closed-loop system in Fig.~\ref{fig:CT_IRC}.}
\label{fig:CT_IRC_equivalence}
\end{figure}

\subsection{Hybrid integrator-gain systems}
A SISO hybrid integrator-gain system (HIGS) $\mathcal{H}$ is represented by the following differential algebraic equations~\cite{deenen2017hybrid}:
	\begin{equation}\label{eq:HIGS}
		\mathcal{H}:
		\begin{cases}
			\dot{x}_h = \omega_h e, & \text{if}\, (e,u,\dot{e}) \in \mathcal{F}_1\\
			x_h = k_he, & \text{if}\, (e,u,\dot{e}) \in \mathcal{F}_2\\
			u = x_h,
		\end{cases}
	\end{equation}
where $x_h,e,u \in \mathbb{R}$ denote the state, input, and output of the HIGS, respectively. Here, $\dot{e}$ is the time derivative of the input $e$, which is assumed to be continuous and piecewise differentiable. Also, $\omega_h \in [0,\infty)$ and $k_h \in (0, \infty)$ represent the integrator frequency and gain value, respectively. These tunable parameters allow for desired control performance. The sets $\mathcal{F}_1$ and $\mathcal{F}_2 \in \mathbb{R}^3$ determine the HIGS modes of operation; i.e. the integrator and gain modes, respectively. The HIGS is designed to operate under the sector constraint $(e,u,\dot{e})\in \mc F$ (see \cite{deenen2017hybrid,Achten_HIGS_Skyhook_thesis_2020}) where
\begin{equation}\label{eq:F}
	\mathcal{F} = \{ (e,u,\dot{e}) \in \mathbb{R}^3 \mid eu \geq \frac{1}{k_h}u^2\},
\end{equation}
and $\mathcal{F}_1$ and $\mathcal{F}_2$ are defined as
\begin{align}
	\mathcal{F}_1& = \mathcal{F} \setminus \mathcal{F}_2;\notag\\
	\mathcal{F}_2& = \{(e,u,\dot{e}) \in \mathbb{R}^3 \mid u = k_he\text{ and } \omega_he^2 > k_he\dot{e}\}.\label{eq:F2}
\end{align}
A HIGS of the form (\ref{eq:HIGS}) is designed to primarily operate in the integrator mode unless the HIGS output $u$ is on the boundary of the sector $\mathcal{F}$, and tends to exit the sector; i.e. $(e,u,\dot{e}) \in \mathcal{F}_2$. In this case, the HIGS is enforced to operate in the gain mode. At the time instants when switching happens, the state $x_h$ still remains continuous, as can be seen from (\ref{eq:HIGS}).

\section{HIGS-based IRC for NI systems}\label{sec:HIGS IRC}
In this section, we provide the system model of a HIGS-based IRC. Also, we show that a HIGS-based IRC has the NI property and can be applied in the control of an NI plant.
\subsection{HIGS-based IRC}
Consider the structure of an IRC as shown in Fig.~\ref{fig:CT_IRC}. A HIGS-based IRC is constructed similarly but with the integrator $C(s)$ in Fig.~\ref{fig:CT_IRC} replaced by a HIGS of the form (\ref{eq:HIGS}). The implementation of a HIGS-based IRC is shown in Fig.~\ref{fig:HIGS-based IRC}.

\begin{figure}[h!]
\centering
\psfrag{r}{$r$}
\psfrag{e_s}{\hspace{0.2cm}$e$}
\psfrag{U_s}{\hspace{0.1cm}$u$}
\psfrag{Y_s}{\hspace{0.3cm}$y$}
\psfrag{barY_s}{\hspace{0.2cm}$\overline y$}
\psfrag{G_s}{$G(s)$}
\psfrag{C_s}{\hspace{-0.03cm}HIGS}
\psfrag{D}{\hspace{0.02cm}$D$}
\includegraphics[width=8.5cm]{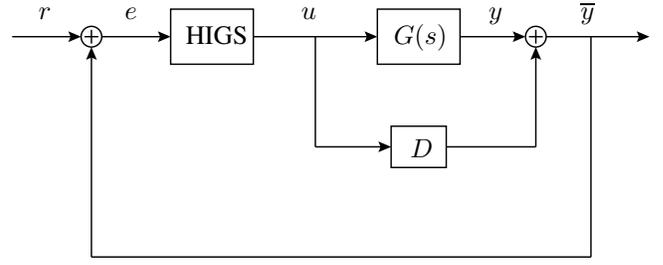}
\caption{Closed-loop interconnection of a HIGS $\mc H$ and $G(s)+D$.}
\label{fig:HIGS-based IRC}
\end{figure}

\begin{figure}[h!]
\centering
\psfrag{r}{$r$}
\psfrag{e_s}{\hspace{0.2cm}$\widetilde e$}
\psfrag{U_s}{\hspace{0.1cm}$u$}
\psfrag{Y_s}{\hspace{0.6cm}$y$}
\psfrag{G_s}{$G(s)$}
\psfrag{K_s}{\hspace{0.2cm}$\widetilde{\mc H}$}
\includegraphics[width=8.5cm]{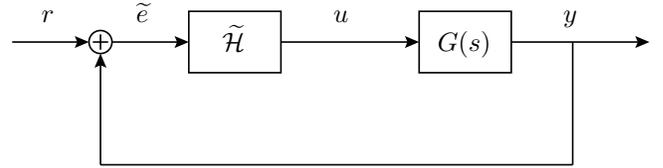}
\caption{Closed-loop interconnection of a HIGS-based IRC and a plant. It is equivalent to the closed-loop system in Fig.~\ref{fig:HIGS-based IRC}.}
\label{fig:HIGS-based IRC equivalence}
\end{figure}

We aim to derive the model of the HIGS-based IRC, which takes $r+y$ as input and gives an output $u$, as shown in Fig.~\ref{fig:HIGS-based IRC equivalence}. According to the settings in Fig.~\ref{fig:HIGS-based IRC} and Fig.~\ref{fig:HIGS-based IRC equivalence}, we have that
\begin{align}
	e=&\ r+y+Du;\notag\\
	\widetilde e=&\ r+y.
\end{align}
Therefore, we have 
\begin{equation}\label{eq:relation between e and tilde e}
	e = \widetilde e+Du=\widetilde e+Dx_h,
\end{equation}
where the second equality uses (\ref{eq:HIGS}). When the HIGS is in the integrator mode, we have that
\begin{equation*}
	\dot x_h=\omega_he=\omega_hDx_h+\omega_h\widetilde e.
\end{equation*}
When the HIGS is in the gain mode, we have that
\begin{equation*}
	x_h=k_he=k_hDx_h+k_h\widetilde e,
\end{equation*}
which implies
\begin{equation*}
	x_h=\frac{k_h}{1-k_hD}\widetilde e.
\end{equation*}
Also, we reformulate the sets $\mc F$, $\mc F_1$ and $\mc F_2$ to pose conditions on $(\widetilde e,u,\dot {\widetilde e})$ instead of $(e,u,\dot e)$.
Substituting (\ref{eq:relation between e and tilde e}) into (\ref{eq:F}) and (\ref{eq:F2}), we have that
\begin{equation*}
	\widetilde {\mc F} = \{ (\widetilde e,u,\dot{\widetilde e}) \in \mathbb{R}^3 \mid \widetilde eu \geq \frac{1-k_hD}{k_h}u^2\},
\end{equation*}
\begin{equation*}
	\widetilde {\mc F}_2 = \{(\widetilde e,u,\dot{\widetilde e}) \in \mathbb{R}^3 \mid u = \frac{k_h}{1-k_hD}\widetilde e \text{ and }  \omega_h \widetilde e^2 > k_h \widetilde e\dot{\widetilde e}\}.
\end{equation*}
To summarize, the system model of a HIGS-based IRC is given as follows:
	\begin{equation}\label{eq:HIGS-based IRC}
		\widetilde {\mathcal{H}}:
		\begin{cases}
			\dot{x}_h = \omega_hDx_h+\omega_h\widetilde e, & \text{if}\, (\widetilde e,u,\dot{\widetilde e}) \in \widetilde {\mathcal{F}}_1\\
			x_h = \widetilde \kappa \widetilde e, & \text{if}\, (\widetilde e,u,\dot{\widetilde e}) \in \widetilde {\mathcal{F}}_2\\
			u = x_h,
		\end{cases}
	\end{equation}
where $x_h,\widetilde e,u\in \mb R$ are the state, input and output of the HIGS-based IRC, respectively. The variable $\dot {\widetilde e}$ denotes the time derivative of the input $\widetilde e$ and is assumed to be continuous and piecewise differentiable. The constants $\omega_h\geq 0$, $D<0$ and $k_h>0$ are the system parameters. Also, we denote the new gain value by
\begin{equation}\label{eq:kappa}
	\widetilde \kappa:=\frac{k_h}{1-k_hD}
\end{equation}
since it will repeatedly occur in what follows. The HIGS-based IRC satisfies the sector constraint $(\widetilde e,u,\dot{\widetilde e})\in \widetilde {\mc F}$ where $\widetilde {\mc F}$, $\widetilde {\mc F}_1$ and $\widetilde {\mc F}_2$ are given as follows:
\begin{align}
	\widetilde {\mc F} =&\ \{ (\widetilde e,u,\dot{\widetilde e}) \in \mathbb{R}^3 \mid \widetilde eu \geq \frac{1}{\widetilde \kappa}u^2\},\label{eq:tilde F}\\
	\widetilde {\mc F}_1=&\ \widetilde {\mc F}\backslash \widetilde {\mc F}_2,\label{eq:tilde F1}\\
	\widetilde {\mc F}_2=&\ \{(\widetilde e,u,\dot{\widetilde e}) \in \mathbb{R}^3 \mid u = \widetilde \kappa\widetilde e \text{ and }  \omega_h \widetilde e^2 > k_h \widetilde e\dot{\widetilde e}\}.\label{eq:tilde F2}
\end{align}

\subsection{NI property of HIGS-based IRC}
In this section, we show that a HIGS-based IRC of the form (\ref{eq:HIGS-based IRC}) is an NI system according to Definition \ref{def:nonlinear NI}.
\begin{theorem}\label{theorem:HIGS-based IRC NI property}
	A HIGS-based IRC of the form (\ref{eq:HIGS-based IRC}) is an NI system with the storage function
	\begin{equation}\label{eq:V_h}
		V_h(x_h)=\frac{1}{2\widetilde \kappa}x_h^2
	\end{equation}
such that
\begin{equation}\label{eq:NI ineq for V_h}
	\dot V_h(x_h)\leq \widetilde e\dot x_h.
\end{equation}
\end{theorem}
\begin{proof}
We have that $\widetilde \kappa>0$ because in (\ref{eq:kappa}) we have $k_h>0$ and $D<0$. Hence, the storage function $V_h(x_h)$ in (\ref{eq:V_h}) is positive definite.
We now show that (\ref{eq:NI ineq for V_h}) is satisfied in both integrator mode and gain mode. We have
\begin{equation}\label{eq:dot V-e dot x}
	\dot V_h(x_h)- \widetilde e\dot x_h = \left(\frac{1}{\widetilde \kappa}x_h-\widetilde e\right)\dot x_h.
\end{equation}
\textbf{\textit{Case 1}}. $(\widetilde e,u,\dot{\widetilde e})\in \widetilde {\mc F}_1$. In this case, we have $\widetilde eu \geq \frac{1}{\widetilde \kappa}u^2$ according to (\ref{eq:tilde F}) and (\ref{eq:tilde F1}). That is $\widetilde ex_h \geq \frac{1}{\widetilde \kappa}x_h^2$ since $u=x_h$. According to the system dynamics in $\widetilde {\mc F}_1$ mode in (\ref{eq:HIGS-based IRC}), we have
\begin{align}
	\dot V_h(x_h)- \widetilde e\dot x_h =&\ \left(\frac{1}{\widetilde \kappa}x_h-\widetilde e\right)(\omega_hDx_h+\omega_h\widetilde e)\notag\\
	=&\ \frac{\omega_h}{\widetilde \kappa}\left(x_h-\widetilde \kappa \widetilde e\right)(Dx_h+\widetilde e).\label{eq:dot V-e dot x in F1 mode}
\end{align}
We discuss in the cases that $x_h=0$ and $x_h\neq 0$. If $x_h=0$, then $\dot V_h(x_h)- \widetilde e\dot x_h=-\omega_h \widetilde e^2\leq 0$. If $x_h\neq 0$, then sector constraint $\widetilde ex_h \geq \frac{1}{\widetilde \kappa}x_h^2$ implies that $\frac{\widetilde e}{x_h}\geq \frac{1}{\widetilde \kappa}$. We can rewrite (\ref{eq:dot V-e dot x in F1 mode}) as
\begin{equation}\label{eq:dot V-e dot x in F1 mode v2}
	\dot V_h(x_h)- \widetilde e\dot x_h=\frac{\omega_h}{\widetilde \kappa}x_h^2\left(1-\widetilde \kappa \frac{\widetilde e}{x_h}\right)\left(D+\frac{\widetilde e}{x_h}\right).
\end{equation}
The right-hand side of (\ref{eq:dot V-e dot x in F1 mode v2}) is the product of a nonnegative quantity $\frac{\omega_h}{\widetilde \kappa}x_h^2$ and a function
\begin{equation*}
	\psi(a)=(1-\widetilde \kappa a)(D+ a),
\end{equation*}
where $a=\frac{\widetilde e}{x_h}\geq \frac{1}{\widetilde \kappa}$. Letting $\psi(a)=0$, we have that $a = \frac{1}{\widetilde \kappa}$ and $a=-D$. Note that $\frac{1}{\widetilde \kappa}=\frac{1}{k_h}-D>-D$. The graph of the function $\psi(a)$ is a parabola that opens downward and intersects with the horizontal axis at $-D$ and $\frac{1}{\widetilde \kappa}$. Therefore, $\psi(a)\leq 0$ for all $a\geq \frac{1}{\widetilde \kappa}$. This implies that $\dot V_h(x_h)- \widetilde e\dot x_h\leq 0$ also holds in the case $x_h\neq 0$.\\
\textbf{\textit{Case 2}}. $(\widetilde e,u,\dot{\widetilde e})\in \widetilde {\mc F}_2$. In this case, we have $x_h=\widetilde \kappa$ according to (\ref{eq:tilde F2}). Therefore, it follows from (\ref{eq:dot V-e dot x}) that $\dot V_h(x_h)-\widetilde e\dot x_h=0$.

Therefore, the condition (\ref{eq:NI ineq for V_h}) is satisfied in both modes. Hence, HIGS-based IRC of the form (\ref{eq:HIGS-based IRC}) is an NI system.
\end{proof}

\subsection{Stability for the interconnection of an NI system and a HIGS-based IRC}
Consider a SISO NI system with the transfer function $G(s)$ and the minimal realization:
\begin{subequations}\label{eq:G(s)}
\begin{align}
\dot x =&\ Ax+Bu,\label{eq:G(s) state equation}\\
y =&\ Cx,
\end{align}
\end{subequations}
where $x\in \mathbb R^n$, $u,y \in \mathbb R$ are the state, input and output of the system, respectively. Here, $A\in \mathbb R^{n\times n}$, $B\in \mathbb R^{n\times 1}$ and $C\in \mathbb R^{1\times n}$. We show that for any NI plant of the form (\ref{eq:G(s)}), there exists a HIGS-based IRC such that their closed-loop interconnection is asymptotically stable.
\begin{figure}[h!]
\centering
\psfrag{r}{\hspace{-0.3cm}$r=0$}
\psfrag{e_s}{\hspace{0.2cm}$\widetilde e$}
\psfrag{U_s}{\hspace{0.1cm}$u$}
\psfrag{Y_s}{\hspace{0.6cm}$y$}
\psfrag{G_s}{$G(s)$}
\psfrag{K_s}{\hspace{0.2cm}$\widetilde{\mc H}$}
\includegraphics[width=8.5cm]{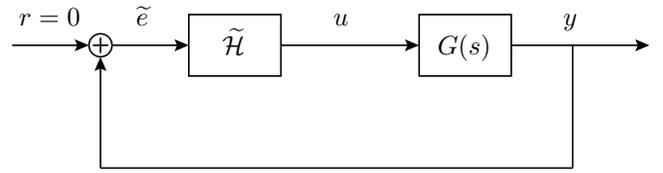}
\caption{Closed-loop interconnection of a HIGS-based IRC and a plant.}
\label{fig:HIGS-based IRC equivalence r=0}
\end{figure}

\begin{theorem}
	Consider the SISO minimal linear NI system (\ref{eq:G(s)}) with transfer function $G(s)$. There exists a HIGS-based IRC $\widetilde{\mc H}$ of the form (\ref{eq:HIGS-based IRC}) such that the closed-loop interconnection of the system (\ref{eq:G(s)}) and the HIGS-based IRC $\widetilde{\mc H}$ as shown in Fig.~\ref{fig:HIGS-based IRC equivalence r=0} is asymptotically stable.
\end{theorem}
\begin{proof}
	Since the system (\ref{eq:G(s)}) is minimal and NI, then according to Lemma \ref{lemma:NI},  $\det A\neq 0$ and there exists a matrix $Y=Y^T>0$, $Y\in \mb R^{n\times n}$ such that
	\begin{equation}\label{eq:LMI conditions}
		AY+YA^T\leq 0, \textnormal{ and }\  B+AYC^T=0.
	\end{equation}
We construct the candidate Lyapunov function of the closed-loop system as follows
\begin{equation*}
	W(x,x_h)=\frac{1}{2}\begin{bmatrix}
		x^T & x_h
	\end{bmatrix}
	\begin{bmatrix}
		Y^{-1}&-C^T\\-C&\frac{1}{\widetilde \kappa}
	\end{bmatrix}\begin{bmatrix}
		x \\ x_h
	\end{bmatrix}
\end{equation*}
Since $Y>0$, then according to Schur complement theorem, the function $W(x,x_h)$ is positive definite if and only if
\begin{equation*}
	\frac{1}{\widetilde \kappa}-CYC^T>0.
\end{equation*}
That is
\begin{equation}\label{eq:kG(0)<1}
	\widetilde \kappa G(0)<1
\end{equation}
considering that $G(0)=-CA^{-1}B=CA^{-1}AYC^T=CYC^T$. We choose the parameters of the HIGS-based IRC to satisfy the condition (\ref{eq:kG(0)<1}).
Also, since $\widetilde \kappa>0$, then $W(x,x_h)$ is positive definite if and only if
\begin{equation}\label{eq:Y-CC}
	Y^{-1}-\widetilde \kappa C^TC>0.
\end{equation}
Therefore, (\ref{eq:kG(0)<1}) is satisfied if and only if (\ref{eq:Y-CC}) is satisfied.
We apply Lyapunov stability theorem in the following. Taking the time derivative of the function $W(x,x_h)$, we have that
\begin{align}
	\dot W&(x,x_h)\notag\\
	 =&\ x^TY^{-1}\dot x+\frac{1}{\widetilde \kappa}x_h\dot x_h-\dot x_hCx-x_hC\dot x\notag\\
	=& \left(x^TY^{-1}-x_hC\right)\dot x+\dot x_h\left(\frac{1}{\widetilde \kappa}x_h-Cx\right)\notag\\
	=& \left(x^TY^{-1}-uC\right)\dot x+\dot x_h\left(\frac{1}{\widetilde \kappa}x_h-\widetilde e\right)\notag\\
	=& \left(x^TY^{-1}+uB^TA^{-T}Y^{-1}\right)\dot x+\dot x_h\left(\frac{1}{\widetilde \kappa}x_h-\widetilde e\right)\notag\\
	=& \left(x^TA^T+uB^T\right)\left(A^{-T}Y^{-1}\right)\dot x+\dot x_h\left(\frac{1}{\widetilde \kappa}x_h-\widetilde e\right)\notag\\
	=&\ \dot x^TA^{-T}Y^{-1}\dot x+\dot x_h\left(\frac{1}{\widetilde \kappa}x_h-\widetilde e\right)\notag\\
	=&\ \frac{1}{2}\dot x^T\left(A^{-T}Y^{-1}+Y^{-1}A^{-1}\right)\dot x+\dot x_h\left(\frac{1}{\widetilde \kappa}x_h-\widetilde e\right),
\end{align}
where $u=x_h$ and $\widetilde e = y =Cx$ are also used. The condition (\ref{eq:LMI conditions}) implies that $A^{-T}Y^{-1}+Y^{-1}A^{-1}\leq 0$. Therefore, $\frac{1}{2}\dot x^T\left(A^{-T}Y^{-1}+Y^{-1}A^{-1}\right)\dot x\leq 0$. Also, $\dot x_h\left(\frac{1}{\widetilde \kappa}x_h-\widetilde e\right)=\dot V_h(x_h)-\widetilde e\dot x_h\leq 0$ as is shown in Theorem \ref{theorem:HIGS-based IRC NI property}. Hence, $\dot W(x,x_h)\leq 0$. The closed-loop system is Lyapunov stable. We apply LaSalle's invariance principle in the following to show that the closed-loop system is indeed asymptotically stable. When $\dot W(x,x_h)$ remains zero, we have both $\left(A^{-T}Y^{-1}+Y^{-1}A^{-1}\right)\dot x$ and $\dot x_h\left(\frac{1}{\widetilde \kappa}x_h-\widetilde e\right)$ remaining zero. The condition $\dot x_h\left(\frac{1}{\widetilde \kappa}x_h-\widetilde e\right)=0$ implies that $\dot x_h=0$ or $x_h=\widetilde \kappa \widetilde e$. We show in the following that the condition $\dot x_h=0$ indeed also implies $x_h=\widetilde \kappa \widetilde e$. We only provide a proof for the case of $\widetilde {\mc F}_1$ mode since $x_h=\widetilde \kappa \widetilde e$ is always true in the $\widetilde {\mc F}_2$ mode. According to the state equation in the $\widetilde {\mc F}_1$ mode given in (\ref{eq:HIGS-based IRC}), we have that $\widetilde e=-Dx_h$. Substituting this equation into the inequality in (\ref{eq:tilde F}), we have that $\left(-D-\frac{1}{\widetilde \kappa}\right)x_h^2\geq 0$. Note that $\left(-D-\frac{1}{\widetilde \kappa}\right)=-\frac{1}{k_h}<0$. Hence, in this case, $x_h=0$. Therefore, $\widetilde e=-Dx_h=0$. Thus, $x_h=0=\widetilde \kappa \widetilde e$. In the case that $x_h\equiv \widetilde \kappa \widetilde e$, we prove that the system cannot stay in the $\widetilde {\mc F}_1$ mode by contradiction. Suppose $x_h\equiv \widetilde \kappa \widetilde e$ and $\dot{x}_h = \omega_hDx_h+\omega_h\widetilde e$. Then, $\widetilde \kappa\dot {\widetilde e}=\omega_hD\widetilde \kappa \widetilde e+\omega_h\widetilde e$. That is $\dot {\widetilde e}=\frac{\omega_h}{1-k_hD}\widetilde e$. This implies that both $y$ and $x_h$ diverge considering that $\frac{\omega_h}{1-k_hD}>0$, $y=\widetilde e$ and $x_h=\widetilde \kappa \widetilde e$. This contradicts the Lyapunov stability of the interconnection as shown above. Therefore, the HIGS-based IRC can only stay in the $\widetilde {\mc F}_2$ mode when $x_h\equiv \widetilde \kappa \widetilde e$. In this case, according to (\ref{eq:tilde F2}), we have
\begin{equation}\label{eq:tilde F2 condition}
	\omega_h\widetilde e^2>k_h \widetilde e \dot{\widetilde e}.
\end{equation}
The condition (\ref{eq:tilde F2 condition}) cannot be satisfied by satisfying $\widetilde e \dot{\widetilde e}<0$ over time because then $\dot V_h(x_h)=\frac{1}{\widetilde \kappa}x_h\dot x_h=\widetilde \kappa \widetilde e \dot{\widetilde e}<0$. This implies that $y=\widetilde e= \frac{1}{\widetilde \kappa}x_h$ converges to zero, which is not the case considered here. Also, in the case that (\ref{eq:tilde F2 condition}) is satisfied with $\dot {\widetilde e}= 0$ overtime, we have that $\dot y=\dot {\widetilde e}\equiv 0$, which implies $\dot x \equiv 0$ according to the observability of the system (\ref{eq:G(s)}). The quantity $\dot x$ is given as follows
\begin{align}
	\dot x=&\ Ax+Bu=Ax+Bx_h=Ax+B\widetilde \kappa \widetilde e=Ax+\widetilde \kappa By\notag\\
	=&\ Ax+\widetilde \kappa BCx=(A+\widetilde \kappa BC)x=(A-\widetilde \kappa AYC^TC)x\notag\\
	=&\ AY(Y^{-1}-\widetilde \kappa C^TC)x.
\end{align}
According to the nonsingularity of the matrices $A$, $Y$ and also (\ref{eq:Y-CC}), $\dot x \neq 0$ for all $x\neq 0$. Therefore, $\dot {\widetilde e}\equiv 0$ implies $x=0$ and $x_h=0$. The system is already at the equilibrium. We have shown that (\ref{eq:tilde F2 condition}) cannot be satisfied with $\widetilde e \dot{\widetilde e}<0$ overtime, and the case $\dot {\widetilde e}=0$ implies the system state is already at the equilibrium. Therefore, we only need to deal with the case that (\ref{eq:tilde F2 condition}) is satisfied with $\widetilde e \dot{\widetilde e}>0$. Since
the trajectories of $\widetilde e$ and $\dot {\widetilde e}$ in the $\widetilde {\mc F}_2$ mode are independent of
$\omega_h$, we can always choose sufficiently small $\omega_h>0$ such that $\omega_h\widetilde e^2<k_h \widetilde e \dot{\widetilde e}$ in order to violate the condition (\ref{eq:tilde F2 condition}). Therefore, the condition $\dot W(x,x_h)\equiv 0$ can be violated by choosing suitable parameters. Then $W(x,x_h)$ will keep decreasing until the system state reaches the origin. This implies that the interconnection in Fig.~\ref{fig:HIGS-based IRC equivalence r=0} is asymptotically stable.
\end{proof}

\section{Proportional-integral-double-integral resonant control}\label{sec:PII2RC}
In this section, we propose another variant of the IRC, where the integrator $C(s)$ in Fig.~\ref{fig:CT_IRC} is replaced by a proportional-integral-double-integral ($\text{PII}^2$) controller with the transfer function
\begin{equation}\label{eq:tilde C(s)}
	\widetilde C(s)=k_p+\frac{k_1}{s}+\frac{k_2}{s^2}.
\end{equation}
Here, $k_p,k_1,k_2\in \mb R$ and we let $k_p,k_1,k_2>0$. Similar to the IRC, a $\text{PII}^2$ resonant controller ($\text{PII}^2\text{RC}$) can be equivalently constructed as shown in Fig.~\ref{fig:CT_IRC_equivalence}, where the transfer function of the $\text{PII}^2\text{RC}$ is
\begin{equation}\label{eq:tilde K(s) in tilde C(s)}
	\widetilde K(s)=\frac{\widetilde C(s)}{1-\widetilde C(s)D}.
\end{equation}
Substituting (\ref{eq:tilde C(s)}) to (\ref{eq:tilde K(s) in tilde C(s)}), we obtain the transfer function of the $\text{PII}^2\text{RC}$:
\begin{equation}\label{eq:tilde K(s)}
	\widetilde K(s)=\frac{k_ps^2+k_1s+k_2}{s^2-k_pDs^2-k_1Ds-k_2D}.
\end{equation}
\begin{lemma}
	The $\text{PII}^2\text{RC}$ system with transfer function $\widetilde K(s)$ given in (\ref{eq:tilde K(s)}) is an SNI system.
\end{lemma}
\begin{proof}
	Since $k_p,k_1,k_2>0$ and $D<0$, then the poles of $\widetilde K(s)$ satisfies $Re[s]<0$. Also,
\begin{align}
	j[\widetilde K(j\omega)-\widetilde K(-j\omega)]=&\frac{2k_1\omega^3}{((1-k_pd)\omega^2+k_2d)^2+(k_1d\omega)^2}\notag\\
	>& 0,
\end{align}
for all $\omega>0$. Therefore, $\widetilde K(s)$ is SNI.
\end{proof}

\begin{theorem}
	Consider the SISO minimal linear NI system (\ref{eq:G(s)}) with transfer function $G(s)$. It can be stabilized using a $\text{PII}^2\text{RC}$ $\widetilde K(s)$ given in (\ref{eq:tilde K(s)}) satisfying $D<-G(0)$.
\end{theorem}
\begin{proof}
	Since $G(s)$ is NI with $G(\infty)=0$ and $\widetilde K(s)$ is SNI, then the interconnection of $G(s)$ and $\widetilde K(s)$ is asymptotically stable if and only if $G(0)\widetilde K(0)<1$. According to (\ref{eq:tilde K(s)}), $K(0)=-\frac{1}{D}$. Therefore, stability is achieved if and only if $G(0)(-\frac{1}{D})<1$. That is, $D<-G(0)$. Note that here $G(0)=-CA^{-1}B=CYC^T\geq 0$, according to Lemma \ref{lemma:NI}. 
\end{proof}

\section{HIGS-based $\text{PII}^2\text{RC}$}\label{sec:HIGS PII2RC}
Motivated by the $\text{PII}^2\text{RC}$ proposed in Section \ref{sec:PII2RC}, we consider a HIGS-based $\text{PII}^2\text{RC}$ in this section. To be specific, we consider replacing the single integrator $k_1/s$ in the $\text{PII}^2\text{RC}$ by a single HIGS of the form (\ref{eq:HIGS}). Also, we replace the double integrator $k_2/s^2$ by two serial cascaded HIGS, both of the form (\ref{eq:HIGS}). 
\begin{figure}[h!]
\centering
\psfrag{r}{\hspace{-0.2cm}$r=0$}
\psfrag{e_s}{\hspace{0.2cm}$e$}
\psfrag{U_s}{\hspace{0.1cm}$u$}
\psfrag{Y_s}{\hspace{0.3cm}$y$}
\psfrag{barY_s}{\hspace{0.2cm}$\overline y$}
\psfrag{G_s}{$G(s)$}
\psfrag{C_s}{\hspace{-0.03cm}HIGS}
\psfrag{D}{\hspace{0.02cm}$D$}
\psfrag{H_1}{\hspace{0.00cm}$\mc H_1$}
\psfrag{H_2}{\hspace{0.00cm}$\mc H_2$}
\psfrag{H_3}{\hspace{0.00cm}$\mc H_3$}
\psfrag{k_p}{\hspace{0.02cm}$k_p$}
\psfrag{x_1}{\hspace{0.02cm}$x_{h1}$}
\psfrag{x_2}{\hspace{0.02cm}$x_{h2}$}
\psfrag{x_3}{\hspace{0.02cm}$x_{h3}$}
\psfrag{y_h}{\hspace{0.02cm}$y_h$}
\psfrag{H_h}{\hspace{0.02cm}$\mc H_h$}
\includegraphics[width=8.5cm]{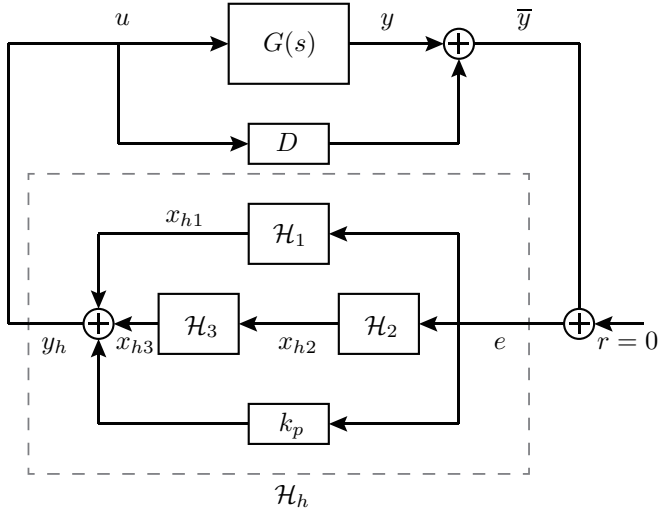}
\caption{Closed-loop interconnection of a HIGS-based $\text{PII}^2$ controller $\mc H_h$ and $G(s)+D$. The system $\mc H_h$ as shown in the dotted line box with input $e$ and output $y_h$ is the parallel cascade of a single HIGS $\mc H_1$, the serial cascade of two HIGS $\mc H_2$ and $\mc H_3$, and also a gain $k_p$.}
\label{fig:HIGS-based PII2RC}
\end{figure}

The HIGS $\mc H_1$, $\mc H_2$ and $\mc H_3$ are of the form (\ref{eq:HIGS}), with different parameters. We provide the system models of $\mc H_1$, $\mc H_2$ and $\mc H_3$ again in the following to distinguish different parameters in these three HIGS:

\begin{equation}\label{eq:HIGS i}
		\mathcal{H}_i:
		\begin{cases}
			\dot{x}_{hi} = \omega_{hi} e_i, & \text{if}\, (e_i,x_{hi},\dot{e_i}) \in \mathcal{F}_{i1}\\
			x_{hi} = k_{hi}e_i, & \text{if}\, (e_i,x_{hi},\dot{e_i}) \in \mathcal{F}_{i2}
		\end{cases}
	\end{equation}
where $e_i \in \mathbb{R}$ is the input, $x_{hi}\in \mb R$ is the state and also the output of the HIGS $\mc H_i$ ($i=1,2,3$), respectively. Here, $\dot{e}_i$ is the time derivative of the input $e_i$, which is assumed to be continuous and piecewise differentiable. The parameters $\omega_{hi} \in [0,\infty)$ and $k_{hi} \in (0, \infty)$ represent the integrator frequency and gain value of the HIGS $\mc H_i$, respectively. Also, we have
\begin{align}
	\mathcal{F}_i &= \{ (e_i,x_{hi},\dot{e}_i) \in \mathbb{R}^3 \mid e_ix_{hi} \geq \frac{1}{k_{hi}}x_{hi}^2\},\label{eq:Fi}\\
	\mathcal{F}_{i1}& = \mathcal{F}_i \setminus \mathcal{F}_{i2};\label{eq:Fi1}\\
	\mathcal{F}_{i2}& = \{(e_i,x_{hi},\dot{e}_i) \in \mathbb{R}^3 \mid x_{hi} = k_{hi}e_i\text{ and } \omega_{hi}e_i^2 > k_{hi}e_i\dot{e}_i\}.\label{eq:Fi2}
\end{align}
According to the setting of the system $\mc H_h$ in Fig.~\ref{fig:HIGS-based PII2RC}, we have that
\begin{equation*}
	e_1=e_2=e;\text{ and } e_3=x_{h2}.
\end{equation*}
The integrator frequency $\omega_{h1}$ of the HIGS $\mc H_1$ corresponds to the parameter $k_1$ of the $\text{PII}^2$ controller $\widetilde C(s)$ given in (\ref{eq:tilde C(s)}). Also, the product of the integrator frequencies $\omega_{h2}$ and $\omega_{h3}$ corresponds to the parameter $k_2$ in (\ref{eq:tilde C(s)}).

We prove in the following that the closed-loop interconnection shown in Fig.~\ref{fig:HIGS-based PII2RC} is asymptotically stable. First, we provide some preliminary results on the nonlinear NI property of a single HIGS and two cascaded HIGS; see also \cite{shi2023MEMS}. Note that the notation used in the present paper is different from that in \cite{shi2023MEMS}.

\begin{lemma}\label{lemma:HIGS NNI}(see \cite{shi2023MEMS,shi2022negative})
A HIGS $\mc H_1$ of the form (\ref{eq:HIGS i}) is a nonlinear NI system from the input $e_1$ to the output $x_{h1}$ with the storage function
	\begin{equation*}
		V_1(x_{h1}) = \frac{1}{2k_{h1}}x_{h1}^2
	\end{equation*}
satisfying
	\begin{equation}\label{eq:theorem_Vdot_HIGS}
	    \dot{V}(x_{h1}) \leq e_1\dot{x}_{h1}.
	\end{equation}
Also, if $\dot{V}(x_{h1}) = e_1\dot{x}_{h1}$ then for all $t\in [t_a,t_b]$ we have that $x_{h1}=k_{h1}e_1$.
\end{lemma}

\begin{lemma}\label{lemma:HIGS lossless condition}(see also \cite{shi2023MEMS})
Consider a HIGS $\mc H_1$ of the form (\ref{eq:HIGS i}). If $\dot{V}(x_{h1}) = e_1\dot{x}_{h1}$, then $x_{h1} = k_{h1}e_1$.
\end{lemma}
\begin{proof}
	This lemma is a special single channel case of Lemma 4 in \cite{shi2023MEMS}.
\end{proof}

For the cascade of the HIGS $mc H_2$ and $\mc H_3$, we assume that
\begin{equation}\label{eq:k2=k3}
	 k_{h2}=k_{h3}, \text{ and } \omega_{h2}<\omega_{h3}.
\end{equation}
This assumption simplifies the proof of Lemma \ref{lemma:cascade HIGS NI}, which is further used Theorem \ref{theorem:HIGS-based IRC NI property}. As Theorem \ref{theorem:HIGS based PII^2RC existence} shows the existence of a stabilizing HIGS-based $\text{PII}^2\text{RC}$ for an NI plant. Therefore, (\ref{eq:k2=k3}) is not a necessary condition for choosing a stabilizing HIGS-based $\text{PII}^2\text{RC}$.

\begin{lemma}\label{lemma:cascade HIGS NI}(see \cite{shi2023MEMS})
	Consider two HIGS $\mc H_2$ and $\mc H_3$ of the form (\ref{eq:HIGS i}) and satisfy (\ref{eq:k2=k3}). Then the serial cascade of $\mc H_2$ and $\mc H_3$ as shown in Fig.~\ref{fig:cascade} is a nonlinear NI system from the input $e_2$ to the output $x_{h3}$ with the storage function
\begin{equation}\label{eq:V for cascade}
V_2(x_{h2},x_{h3}) = \frac{1}{2}x_{h2}^2
\end{equation}
satisfying
\begin{equation}\label{eq:NI ineq cascade}
\dot V(x_{h2},x_{h3}) \leq e_2 \dot x_{h3}.	
\end{equation}
Moreover, if $\dot V(x_{h2},x_{h3})=e_2\dot x_{h3}$ over a time interval $[t_a,t_b]$, where $t_a<t_b$, then for all $t\in [t_a,t_b]$ we have that $x_{h2} = k_{h2}e_2$ and $x_{h3} = k_{h3}x_{h2}=k_{h2}^2e_2$. 
\end{lemma}
\begin{proof}
	The proof follows directly from Theorem 5 in \cite{shi2023MEMS}, with (\ref{eq:k2=k3}) assumed. The parameter $a$ in Theorem 5 in \cite{shi2023MEMS} is allowed to take the value $a = \frac{k_{h3}}{2k_{h2}}=\frac{1}{2}$, which results the storage function $V_2(x_{h2},x_{h3})$ in (\ref{eq:V for cascade}) to be positive semidefinite instead of positive definite. Regarding the analysis of the case where $\dot V(x_{h2},x_{h3})=e_2\dot x_{h3}$ over a finite time interval, the proof of Theorem 5 in \cite{shi2023MEMS} still remains valid under the assumptions (\ref{eq:k2=k3}) and $a = \frac{1}{2}$, with slight modifications.
\end{proof}

\begin{figure}[h!]
\centering
\psfrag{e_2}{$e_2$}
\psfrag{e_3}{$e_3$}
\psfrag{H_2}{$\mc H_2$}
\psfrag{H_3}{$\mc H_3$}
\psfrag{x_2}{$x_{h2}$}
\psfrag{x_3}{$x_{h3}$}
\includegraphics[width=8.5cm]{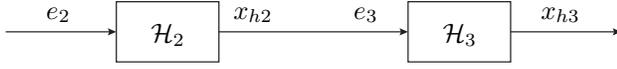}
\caption{The serial cascade of two HIGS $\mc H_2$ and $\mc H_3$}
\label{fig:cascade}
\end{figure}

\begin{theorem}\label{theorem:HIGS based PII^2RC existence}
	Consider the SISO minimal linear NI system (\ref{eq:G(s)}) with transfer function $G(s)$. Also, consider a HIGS-based $\text{PII}^2$ controller $\mc H_h$ applied in positive feedback with $G(s)+D$, as shown in Fig.~\ref{fig:HIGS-based PII2RC}, where $D<0$ is a scalar. Here, $\mc H_h$ is the parallel cascade of the HIGS $\mc H_1$, the serial cascade of two HIGS $\mc H_2$ and $\mc H_3$, and also a gain $k_p>0$. Each  HIGS is of the form (\ref{eq:HIGS i}), where (\ref{eq:k2=k3}) is also assumed. Then there exists a set of parameters $\{D,k_p,\omega_{h1},\omega_{h2},k_{h1},k_{h2}\}$ such that the closed-loop system shown in Fig.~\ref{fig:HIGS-based PII2RC} is asymptotically stable.
\end{theorem}
\begin{proof}
According to Fig.~\ref{fig:HIGS-based PII2RC}, we have that
\begin{align}
	e =&\ Cx+Du= Cx+D(x_{h1}+x_{h3}+k_pe)\notag\\
	\implies e =&\ \gamma Cx+\gamma D(x_{h1}+x_{h3}),\label{eq:e in x}
\end{align}
where
\begin{equation}\label{eq:gamma}
	\gamma = \frac{1}{1-Dk_p}>0.
\end{equation}
Also, we have that
\begin{align}
	u =&\ x_{h1}+x_{h3}+k_pe\notag\\
	=&\ k_p\gamma Cx+\gamma (x_{h1}+x_{h3})\label{eq:u in x}
\end{align}

	We apple Lyapunov's direct method using the candidate Lyapunov function
	\begin{equation}\label{eq:W}
		W(x,x_{h1},x_{h2},x_{h3})=\frac{1}{2}\begin{bmatrix}
			x^T&x_{h1}&x_{h2}&x_{h3}
		\end{bmatrix}M \begin{bmatrix}
			x\\x_{h1}\\x_{h2}\\x_{h3}
		\end{bmatrix}
	\end{equation}
where	
	\begin{equation*}
M=\begin{bmatrix}
			Y^{-1}-k_p\gamma C^TC & -\gamma C^T & 0 & -\gamma C^T\\
			-\gamma C & \frac{1}{k_{h1}}-D\gamma & 0 & -D\gamma \\
			0 & 0 & 1 & 0\\
			-\gamma C & -D\gamma & 0 & -D\gamma
		\end{bmatrix}.
	\end{equation*}
First, we show the positive definiteness of the function $W(x,x_{h1},x_{h2},x_{h3})$ by showing $M>0$. Observing the third block row and block column of $M$, we have $M>0$ if and only if
\begin{equation*}
	\widetilde M = \begin{bmatrix}
			Y^{-1}-k_p\gamma C^TC & -\gamma C^T & -\gamma C^T\\
			-\gamma C & \frac{1}{k_{h1}}-D\gamma & -D\gamma \\
			-\gamma C & -D\gamma & -D\gamma
		\end{bmatrix}>0.
\end{equation*}
Since $k_p>0$ and $D<0$, we have that $\gamma>0$. Therefore, $\widetilde M>0$ if and only if
\begin{align}
	\widehat M = \frac{1}{\gamma}\widetilde M=\begin{bmatrix}
			\frac{1}{\gamma} Y^{-1}-k_p C^TC & - C^T & - C^T\\
			- C & \frac{1}{k_{h1}\gamma}-D & -D \\
			- C & -D & -D
		\end{bmatrix}> 0
\end{align}
We apply Schur complement theorem in the following. Observing that the $3,3$ block $\widehat M_{33}=-D>0$, then $\widehat M>0$ if and only if
\begin{align}
	\overline M =& \widehat M/\widehat M_{33}\notag\\
	=&\begin{bmatrix}
			\frac{1}{\gamma} Y^{-1}-k_p C^TC & - C^T\\
			- C & \frac{1}{k_{h1}\gamma}-D \end{bmatrix}+\frac{1}{D}\begin{bmatrix}
				C^T\\ D
			\end{bmatrix}\begin{bmatrix}
				C & D
			\end{bmatrix}\notag\\
			=&\begin{bmatrix}
				\frac{1}{\gamma} Y^{-1}-(k_p-\frac{1}{D}) C^TC & 0 \\ 0 & \frac{1}{k_{h1}\gamma}
			\end{bmatrix}\label{eq:bar M}
\end{align}
Since $\frac{1}{k_{h1}\gamma}>0$, then we have $\overline M>0$ if and only if
\begin{equation}\label{eq:bar M>0 condition}
	\frac{1}{\gamma} Y^{-1}-(k_p-\frac{1}{D}) C^TC>0
\end{equation}
Substituting (\ref{eq:gamma}) into (\ref{eq:bar M>0 condition}), we have that (\ref{eq:bar M>0 condition}) holds if and only if
\begin{equation}\label{eq:bar M>0 condition2}
	Y^{-1}+\frac{1}{D}C^TC>0
\end{equation}
Using Schur complement theorem, (\ref{eq:bar M>0 condition2}) is equivalent to the positive definiteness of the matrix
\begin{equation}\label{eq:Q}
	Q = \begin{bmatrix}
		Y^{-1} & C^T\\ C & -D
	\end{bmatrix}
\end{equation}
Also, considering that $Y>0$, we have $Q>0$ if and only if 
\begin{equation}\label{eq:-D-CYC^T>0}
	-D-CYC^T>0.
\end{equation}
Considering that $G(0)=-CA^{-1}B=CYC^T$ according to (\ref{eq:NI lemma}), the condition (\ref{eq:-D-CYC^T>0}) can be expressed as
\begin{equation}\label{eq:D<-G(0)}
	D<-G(0).
\end{equation}
Therefore, the function $W(x,x_{h1},x_{h2},x_{h3})$ given in (\ref{eq:W}) is positive definite if and only if (\ref{eq:D<-G(0)}) is satisfied. Taking the time derivative of $W(x,x_{h1},x_{h2},x_{h3})$, we have
\begin{align}
	\dot W&(x,x_{h1},x_{h2},x_{h3})\notag\\
	=&\ x^TY^{-1}\dot x+\frac{1}{k_{h1}}x_{h1}\dot x_{h1}+x_{h2}\dot x_{h2}-\gamma (\dot x_{h1}+\dot x_{h3})Cx\notag\\
	&-\gamma (x_{h1}+x_{h3})C\dot x-\gamma D(x_{h1}+x_{h3})(\dot x_{h1}+\dot x_{h3})\notag\\
	&-k_p\gamma x^TC^TC\dot x\notag\\
	=&\left[x^TY^{-1}-\left(\gamma (x_{h1}+x_{h3}) -k_p\gamma x^TC^T\right) \right]C\dot x\notag\\
	&+\left[\frac{1}{k_{h1}}x_{h1}\dot x_{h1}-(\gamma Cx+\gamma D(x_{h1}+x_{h3}))\dot x_{h1}\right]\notag\\
	&+\left[x_{h2}\dot x_{h2}-(\gamma Cx+\gamma D(x_{h1}+x_{h3}))\dot x_{h3}\right]\notag\\
	=&\left(x^TY^{-1}-uC\right)\dot x+\left(\dot V_1(x_{h1})-e\dot x_{h1}\right)\notag\\
	&+ \left(\dot V_2(x_{h2},x_{h3})-e\dot x_{h3}\right).
\end{align}
We have that
\begin{align}
	\left(x^TY^{-1}-uC\right)\dot x=& \left(x^TA^TA^{-T}Y^{-1}+uB^TA^{-T}Y^{-1}\right)\dot x\notag\\
	=&\left(x^TA^T+uB^T\right)A^{-T}Y^{-1}\dot x\notag\\
	=&\ \dot x^T(A^{-T}Y^{-1})\dot x\notag\\
	=&\ \frac{1}{2}\dot x^T(A^{-T}Y^{-1}+Y^{-1}A^{-1})\dot x\notag\\
	\leq &\ 0.
\end{align}
Also, we have that $\dot V_1(x_{h1})-e\dot x_{h1}\leq 0$ and $\dot V_2(x_{h2},x_{h3})-e\dot x_{h3}\leq 0$ according to the NI property of the $\mc H_1$ and the cascade of $\mc H_2$ and $\mc H_3$, as shown in Lemmas \ref{lemma:HIGS NNI} and \ref{lemma:cascade HIGS NI}. Therefore, $\dot W(x,x_{h1},x_{h2},x_{h3})\leq 0$, which implies that the closed-loop interconnection in Fig.~\ref{fig:HIGS-based PII2RC} is Lyapunov stable. We apply LaSalle's invariance principle in the following to show asymptotic stability. In the case that $\dot W(x,x_{h1},x_{h2},x_{h3})$ remains zero, we have that $\dot x^T(A^{-T}Y^{-1}+Y^{-1}A^{-1})\dot x\equiv 0$, $\dot V_1(x_{h1})-e\dot x_{h1}\equiv 0$ and $\dot V_2(x_{h2},x_{h3})-e\dot x_{h3}\equiv 0$. According to Lemma \ref{lemma:HIGS lossless condition}, $\dot V_1(x_{h1})-e\dot x_{h1}\equiv 0$ implies 
\begin{equation}\label{eq:x1=k1e}
	x_{h1}\equiv k_{h1}e.
\end{equation}
Also, according to Lemma \ref{lemma:cascade HIGS NI}, $\dot V_2(x_{h2},x_{h3})-e\dot x_{h3}\equiv 0$ implies
\begin{align}
	x_{h2}\equiv &\ k_{h2}e,\label{eq:x2=k2e}\\
	x_{h3}\equiv &\ k_{h2}^2e.\label{eq:x3=k3e}
\end{align}
We show that the HIGS $\mc H_1$, $\mc H_2$ and $\mc H_3$ cannot stay in the integrator mode $\mc F_{i1}$ by contradiction. According to (\ref{eq:HIGS i}), if $(e_i,x_{hi},\dot{e_i}) \in \mathcal{F}_{i1}$, then $\dot x_{hi}=\omega_{hi}e_i$. Since $x_{hi}=k_{hi}e_i$, then we have $k_{hi}\dot e_i=\omega_{hi}e_i$. That is,
\begin{equation*}
	\dot e_i=\frac{\omega_{hi}}{k_{hi}}e_i,
\end{equation*}
which implies that $e_i$ diverges and so is $x_{hi}$. This contradicts the Lyapunov stability of the interconnection that is proved above. Therefore, the HIGS $\mc H_1$, $\mc H_2$ and $\mc H_3$ all stay in the gain mode $\mc F_{i2}$. In this case, according to (\ref{eq:Fi2}), we have that $\omega_{hi}e_i^2>k_{hi}e_i\dot e_i$. That is
\begin{align}
	\omega_{h1}e^2> k_{h1}e\dot e;\notag\\
	\omega_{h2}e^2> k_{h2}e\dot e;\notag\\
	\omega_{h2}x_{h2}^2> k_{h2}x_{h2}\dot x_{h2}\implies  \omega_{h2}e^2> k_{h2}e\dot e,\notag
\end{align}
for $\mc H_1$, $\mc H_2$ and $\mc H_3$, respectively. Hence, we have
\begin{equation}\label{eq:rho}
	\rho e^2> e\dot e,
\end{equation}
where $\rho = \min\{\frac{\omega_{h1}}{k_{h1}}, \frac{\omega_{h2}}{k_{h2}}\}$. We show in the following that the condition (\ref{eq:rho}) can be satisfied over time by satisfying $e\dot e<0$. In this case that $e\dot e<0$, the HIGS input $e$ converges. This implies that $x_{h1}$, $x_{h2}$ and $x_{h3}$ all converge to zero. Also, according to (\ref{eq:e in x}) and (\ref{eq:u in x}), $y$ and $u$ also converge. This is not the case of $\dot W(x,x_{h1},x_{h2},x_{h3})\equiv 0$ that is considered here. Also, we can avoid the case that (\ref{eq:rho}) is satisfied by satisfying $\dot e\equiv 0$ overtime. When $e$ is a constant, the HIGS states $x_{h1}$, $x_{h2}$ and $x_{h3}$ are all constants, according to (\ref{eq:x1=k1e}), (\ref{eq:x2=k2e}) and (\ref{eq:x3=k3e}). Also, according to (\ref{eq:e in x}) and (\ref{eq:u in x}), we have that $y=Cx$ is a constant and also $u$ is a constant. Since the system (\ref{eq:G(s)}) is observable, we have that $\dot x=0$ and the system (\ref{eq:G(s)}) is in a steady state. Denote the constant values of $u$ and $e$ by $\overline u$ and $\overline e$, respectively, we have that
\begin{equation*}
	\overline u = (k_{h1}+k_{h2}^2+k_p)\overline e,
\end{equation*}
according to the system setting in Fig.~\ref{fig:HIGS-based PII2RC}. Also, we have that 
\begin{equation*}
	\overline e = (G(0)+D)\overline u.
\end{equation*}
Therefore, by choosing suitable parameters $k_{h1}$, $k_{h2}$, and $k_p$ such that
\begin{equation*}
	k_{h1}+k_{h2}^2+k_p\neq \frac{1}{G(0)+D},
\end{equation*}
we can avoid the case that $\dot e\equiv 0$. Hence, $e\dot e>0$ will be satisfied eventually. In this case, since the trajectories of $e$ and $\dot e$ are independent of $\omega_{h1}$ and $\omega_{h2}$ when all of the HIGS are in the gain mode $\mc F_{2i}$, then we can always choose $\omega_{h1}$ or $\omega_{h2}$ to be sufficiently small such that $\rho e^2<e\dot e$. Therefore, $\dot W(x,x_{h1},x_{h2},x_{h3})$ cannot remain zero over time and $W(x,x_{h1},x_{h2},x_{h3})$ will keep decreasing until $W(x,x_{h1},x_{h2},x_{h3})=0$. This implies that the interconnection in Fig.~\ref{fig:HIGS-based PII2RC} is asymptotically stable.
\end{proof}

\section{Example}\label{sec:example}
In this section, we apply the proposed HIGS-based IRC to stabilize a mass-spring system. As is shown in Fig.~\ref{fig:example}, the mass of the cart is $m=1kg$ and the spring constant is $k=1N/m$. The state-space model of the system is given as follows:
\begin{equation}\label{eq:example}
	\begin{aligned}
		\dot x =& \begin{bmatrix}
			0&1\\-1&0
		\end{bmatrix}x+\begin{bmatrix}
			0\\1
		\end{bmatrix}u;\\
		y=& \begin{bmatrix}
			1&0
		\end{bmatrix}x,
	\end{aligned}
\end{equation}
where $x=\begin{bmatrix}
	x_1\\x_2
\end{bmatrix}$ is the system state with $x_1$ and $x_2$ being its displacement and velocity, respectively. Also, $u$ is an external force input and we measure the system displacement as its output $y$.
\begin{figure}[h!]
\centering
\ctikzset{bipoles/length=1cm}
\begin{circuitikz}
	\pattern[pattern = north east lines] (-0.2,-0.2) rectangle (0,1.5);
	\draw[thick] (0,0) -- (0,1.5);
	\pattern[pattern = north east lines] (0,-0.2) rectangle (5,0);
	\draw[thick] (0,0) -- (5,0);	
	\draw[thick] (0,0) -- (0,1.5);
	\draw (0,0.75) to [spring, l=${k=1N/m}$] (2,0.75);
	\draw[fill=gray!40] (2,0) rectangle (4,1.5);
\node at (3,0.75){$m=1kg$};
\draw[thick, ->] (4,0.75) -- (5,0.75);
        \node at (5.3,0.75){$u$};
 \draw[thick, |->] (3,2) -- (4,2);
        \node at (4.3,2){$y$};
\end{circuitikz}
\caption{A mass-spring system with mass $m = 1kg$ and spring constant $k=1N/m$.}
\label{fig:example}
\end{figure}
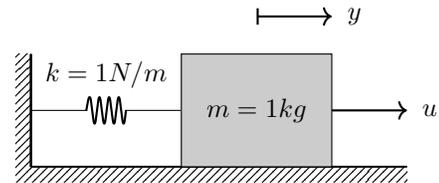
The system (\ref{eq:example}) has a transfer function $G(s)=\frac{1}{s^2+1}$.  We construct a HIGS-based IRC of the form (\ref{eq:HIGS-based IRC}) with
\begin{equation}\label{eq:HIRC parameters}
	\omega_h=0.5,\ k_h=20,\ D=-1.
\end{equation}
Using (\ref{eq:kappa}), we have that $\widetilde \kappa = \frac{5}{6}$. Such a $\widetilde \kappa$ satisfies the condition $\widetilde \kappa G(0)<1$. We set the initial state of the system (\ref{eq:example}) be $x_1(0) = 3$, $x_2(0)=1$. The initial state of the HIGS-based IRC is set to be zero. We can see from Fig.~\ref{fig:simulation} that the states of the system (\ref{eq:example}) converge to the origin under the effect of the HIGS-based IRC.

\begin{figure}[h!]
\centering
\psfrag{Time}{\hspace{-0.25cm}Time ($s$)}
\psfrag{xa}{\scriptsize$x_1$}
\psfrag{xb}{\scriptsize$x_2$}
\psfrag{xho}{\scriptsize$x_h$}
\psfrag{State}{State}
\psfrag{State Trajectories}{\hspace{-0.2cm}State Trajectories}
\psfrag{5}{\scriptsize$5$}
\psfrag{4}{\scriptsize$4$}
\psfrag{2}{\scriptsize$2$}
\psfrag{-2}{\hspace{-0.15cm}\scriptsize$-2$}
\psfrag{0}{\scriptsize$0$}
\psfrag{10}{\scriptsize$10$}
\psfrag{15}{\scriptsize$15$}
\psfrag{20}{\scriptsize$20$}
\psfrag{25}{\scriptsize$25$}
\psfrag{30}{\scriptsize$30$}
\psfrag{35}{\scriptsize$35$}
\psfrag{40}{\scriptsize$40$}
\includegraphics[width=9cm]{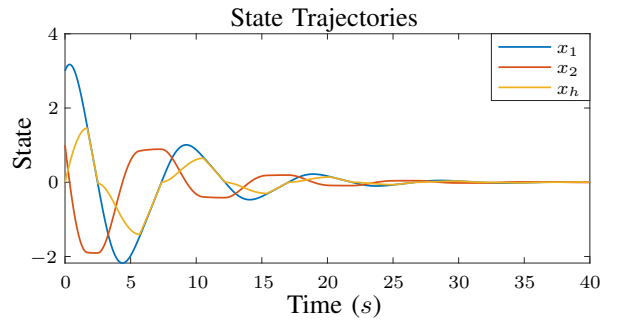}
\caption{State trajectories for the feedback interconnection of the system (\ref{eq:example}) and a HIGS-based IRC with parameters (\ref{eq:HIRC parameters}).}
\label{fig:simulation}
\end{figure}

\section{Conclusion}\label{sec:conclusion}
In this paper, we introduce a HIGS-based IRC to provide a control approach for an NI system, with the advantages of both HIGS and IRC utilized. A HIGS-based IRC is achieved by replacing the integrator in an IRC by a HIGS element. We show that a HIGS-based IRC is an NI system and can stabilize an NI plant when applied in positive feedback. Also, we propose a $\text{PII}^2\text{RC}$ and a HIGS-based $\text{PII}^2\text{RC}$ for the control of NI systems. We show that both a $\text{PII}^2\text{RC}$ and a HIGS-based $\text{PII}^2\text{RC}$ can provide asymptotic stabilization for an NI system. An illustrative example is also provided.

\bibliographystyle{IEEEtran}

\end{document}